\documentclass[journal,onecolumn]{IEEEtran}
\usepackage{generic}
\usepackage{cite}
\usepackage{amsmath,amssymb,amsfonts}
\usepackage{algorithmic}
\usepackage{graphicx}
\usepackage{textcomp}
\usepackage{amsfonts,amssymb,amsmath}
\usepackage{graphicx,graphics,epsfig,color}
\usepackage{multicol}
\usepackage{float}
\usepackage{subcaption}
\usepackage{comment}
\usepackage{caption}

\usepackage{graphicx} 
\usepackage{here}
\usepackage{amsmath}
\usepackage{amsthm}
\usepackage{amssymb}
\usepackage{amsfonts}
\usepackage{graphicx}
\usepackage{subcaption}
\usepackage{lipsum}
\usepackage{xpatch}
\xpatchcmd{\paragraph}{\normalfont}{{\normalfont\bfseries}}{}{}
\usepackage{color}
\usepackage{dsfont}

\usepackage{nicematrix}

\usepackage[ruled,vlined,linesnumbered]{algorithm2e}

\newtheorem{theorem}{Theorem}
\newtheorem{remark}{Remark}
\newtheorem{lemma}{Lemma}
\newtheorem{proposition}{Proposition}

\newtheorem{definition}{Definition}

\newtheorem{problem}{Problem}

\newcommand{\R}{{\mathbb{R}}}

\newcommand{\N}{{\mathbb{N}}}

\newcommand{\He}{\textrm{He}}

\newcommand{\eig}{\textrm{eig}}
\newcommand{\diag}{{\mathrm{Diag}}}
\newcommand{\row}{{\mathrm{Row}}}

\newcommand{\dom}{\mathop{\rm dom}\nolimits}

\usepackage{hyperref}

\usepackage{enumitem}
\usepackage{balance}

\definecolor{olivegreen}{rgb}{0.14,0.29,0}

\newif\ifitsdraft
\def\itsdraft{\global\itsdrafttrue}

\itsdraft  

\ifitsdraft

\definecolor{gray}{rgb}{0.33,0.4,0.47}

\definecolor{steelblue}{rgb}{0,.42,.7}

\definecolor{britishgreen}{rgb}{0,0.26,0.15}
\definecolor{navyblue}{rgb}{0,0,.8}
\definecolor{olivegreen}{rgb}{0.14,0.29,0}
\definecolor{myred}{rgb}{0.86,0.1,0.16}
    
\newcounter{al}    \newcounter{ss}

    \allowdisplaybreaks
    \pdfoutput = 1
    
     \usepackage{soul} 

\usepackage{rgsEnvironments}
\usepackage{hslTR}

\begin{document}
\title{Robust Hybrid Finite Time Parameter Estimation Without Persistence of Excitation}
\author{Adnane Saoud, Ryan  S. Johnson, and Ricardo G. Sanfelice, 
\thanks{A. Saoud is with the College of Computing, University Mohammed VI Polytechnic, Benguerir, Morocco {e-mail:adnane.saoud@um6p.ma}). Ryan S. Johnson and Ricardo  G.  Sanfelice are with the Dept. of Electrical and Computer Engineering, University of California,  Santa Cruz,  CA,  USA (e-mail:ricardo@ucsc.edu,rsjohnso@ucsc.edu). \\ Research by R. G. Sanfelice partially supported by NSF Grants no. CNS-2039054 and CNS-2111688, by AFOSR Grants nos. FA9550-23-1-0145, FA9550-23-1-0313, and FA9550-23-1-0678, by AFRL Grant nos. FA8651-22-1-0017 and FA8651-23-1-0004, by ARO Grant no. W911NF-20-1-0253, and by DoD Grant no. W911NF-23-1-0158.}}

\maketitle

\begin{abstract} In this paper, we consider the problem of estimating parameters of a linear regression model. Using a hybrid systems framework, a hybrid algorithm is proposed allowing the estimate to converge to the exact value of the unknown parameters in predetermined finite time. Interestingly, we show that for the case of constant parameters, the convergence property holds while only requiring the regressor to be exciting on a given interval. For the case of piecewise constant parameters, the classical persistency of excitation condition is required to guarantee the convergence. Robustness of the proposed algorithm with respect to measurements noise is analysed. Finally, illustrative examples are provided showing the merits of the proposed approach in terms of scalability and the applicability for the general class of time-varying unknown parameters.\end{abstract}


\section{INTRODUCTION}

Accurate estimation of model parameters of a system is critical in most applications. Different algorithms have been proposed in the adaptive control community to tackle this problem~\cite{narendra2012stable,tao2003adaptive,saoud2024hybrid}. In static linear regression models~\cite{narendra2012stable,tao2003adaptive}, the relationship between the regressors and the output is linear, {and the estimation of the parameters is generally based on the gradient descent algorithm~\cite{narendra2012stable,tao2003adaptive}, which requires a persistence of excitation (PE) condition~\cite{tao2003adaptive,narendra1987persistent} to ensure exponential convergence}. Regarding finite-time convergence, several results have been proposed recently including~\cite{wang2019robust,wang2019fixed,chowdhary2013concurrent,glushchenko2023unknown,marino2022exponentially}.

Motivated by the results on finite time observers~\cite{engel2002continuous}, this paper presents a hybrid estimator with predetermined time for convergence of the estimates to unknown parameters. We first study the stability and finite time convergence properties of our hybrid estimator for the case of a constant unknown parameter. Interestingly, we show that the regressor needs to be exciting only over a bounded interval of time. Then we show how this result can be generalized to piecewise constant unknown parameters. Finally, we use tools developed in hybrid systems theory to provide robustness of the proposed estimator, where we analyse the convergence of the proposed hybrid estimator under measurement noises in terms of input to state stability. Indeed, we show that under a persistency of excitation type condition one can ensure the input to state stability of the estimator with respect to measurement noises, while a more relaxed finite time excitation condition makes it possible to establish an integral input to state stability property.

In spirit, our approach is closely related to the one in~\cite{hartman2012robust}, which differs from ours in two directions: it deals with dynamical systems rather than with an algebraic input-output model, and it relies on a persistency of excitation condition to ensure that the hybrid system is well defined and that solutions are complete, whereas we only need the regressor to be exciting on a finite time interval. {Several other methods achieve finite-time parameter convergence under relaxed excitation conditions. Memory-based methods achieve finite-time convergence under finite excitation by storing informative regressor data: memory-augmented identification~\cite{vahidi2021memory} and concurrent learning~\cite{tatari2021finite} guarantee convergence when stored samples satisfy a rank condition, avoiding persistent excitation requirements. A second family of methods exploits the Dynamic Regressor Extension and Mixing (DREM) framework introduced by Aranovskiy \emph{et al.}~\cite{aranovskiy2016performance}. DREM methods augment the regressor via dynamic operators and mix extended regressors to obtain decoupled scalar problems, enabling exponential (or finite-time) convergence under PE or interval PE~\cite{aranovskiy2016performance,ortega2020convergence}. Predefined-time variants of DREM  have been proposed to guarantee convergence under finite excitation~\cite{cui2024predefined}. When classical PE is not satisfied but the regressor still contains valuable directional information, \cite{marino2022exponentially} proposes a characterization of partial excitation ensuring exponentially convergent parameter estimation. Similar ideas have been explored in~\cite{aranovskiy2023drem,glushchenko2023relaxation} to address the lack of finite-time PE under the DREM framework. Finally, the approach in~\cite{ortega2021parameter} addressed the alertness problem in DREM-based estimators through a sliding window approach, which maintains tracking capability. However, DREM's transformation from a linear input-output model to $n$ scalar systems requires computing the determinant on inverse of matrices at $O(n^3)$ cost per timestep~\cite{aranovskiy2016performance,ortega2020convergence,wang2019robust}, which is a computational bottleneck inherent to all DREM-based approaches. Our hybrid estimator avoids this overhead by operating directly in $n$-dimensional parameter space with $O(n)$ cost of a simple gradient during flows and $O(n^3)$ operations only at sparse jumps for state transition matrix computation and inversion, making it practical for high-dimensional systems (please see Section \ref{sec:6}). Alternative finite-time methods for parameter estimation include the work in~\cite{ortega2002line}, which achieves finite-time convergence via continuous singular value decomposition, the work in~\cite{adetola2008finite}, which requires specific initialization and bounded parameters. Our hybrid approach thus achieves finite-time convergence without requiring \emph{a priori} parameter bounds, with $O(n)$ flow cost and $O(n^3)$ operations only at sparse jumps, and with formal robustness guarantees.

A preliminary version of this work has been presented in the conference paper~\cite{saoud2021robust}. In the current version, we present a detailed description of the results in~\cite{saoud2021robust} including all the proofs. Moreover, while in~\cite{saoud2021robust} we have shown the input state stability of the perturbed hybrid estimator under a persistency of excitation condition, we first improve the result in \cite{saoud2021robust} by providing tighter input to state stability bounds. Moreover, we also show in this paper that the perturbed hybrid estimator is integral input to state stable if the regressor is exciting only on a finite time interval.

The remainder of this paper is organized as follows. Section~\ref{sec:2} introduces the required preliminaries. In Section~\ref{sec:4}, we show the finite time convergence of the proposed estimator with respect to constant and piecewise constant unknown parameters. In Section~\ref{sec:5}, the robustness of the proposed hybrid estimator to generic noise is investigated. Finally, Section~\ref{sec:6} presents numerical results validating and highlighting the robustness and scalability properties of the proposed estimator.
\section{Preliminaries}
\label{sec:2}
\textbf{Notations:} The symbol $|.|$ denotes the Euclidean norm. Given a matrix $A$, $\eig(A)$ denotes its eigenvalues, $\overline{\mu}(A):=\max\{\lambda/2 : \lambda \in \eig(A+A^T)\}$ and $\underline{\mu}(A):=\min\{\lambda/2 : \lambda \in \eig(A+A^T)\}$. Given a set-valued map $F:\R^n\rightrightarrows \R^m$, its domain is given by $\dom F =\{x\in \R^n : F(x)\neq \emptyset\}$. A continuous function $\alpha$ is said to belong to class $\mathcal{K}_{\infty}$ if it is strictly increasing, $\alpha(0)=0$ and $\alpha(r)$ goes to infinity as $r$ tends to infinity. Given $y\in \R^n$ and a non-empty set $\mathcal{A} \subseteq \R^n$, $|y|_{\mathcal{A}}:= \inf\limits_{x\in \mathcal{A}}|x-y|$. Given matrices $A$ and $B$, $\He(A,B)=A^{\top}B + BA$. Given vectors $A_1,A_2,\ldots,A_n$, $\row(A_1,A_2,\ldots,A_n) \in \mathbb{R}^{nM}$ denotes the vector consisting of the concatenation of the vectors $A_i$, $i=1,2,\ldots,n$. 
\subsection{Linear Regression}
A general linear regression model is given by
 \begin{equation}
 \label{eqn:estimator}
     y(t)=\theta^{*\top}\phi(t)
 \end{equation}
 where $y(t)\in \R$ is the known output, $\theta^* \in \R^n$ is the unknown parameter, and $\phi(t) \in \R^n$ is the known regressor. For the estimator $\hat{y}(t)=\theta^{\top}(t)\phi(t)$, the parameter estimation error is given by $\Tilde{\theta}(t):=\theta(t)-\theta^*$.  The error between the estimated and actual outputs is given by $e(t):=\hat{y}(t)-y(t)=\Tilde{\theta}^{\top}(t)\phi(t)$. The classical gradient descent algorithm~\cite{narendra2012stable,tao2003adaptive} gives $\dot{\theta}=-\gamma\phi(\phi^{\top}\theta-y)$ with $\gamma>0$. It can be shown using Barbalat's Lemma~\cite{tao2003adaptive} that using the gradient descent algorithm $e(t)$ converges to $0$. However, the exponential convergence of $\theta$ to $\theta^*$ requires persistent excitation~\cite{tao2003adaptive,narendra1987persistent} (Definition~\ref{def:persistency2}).

 {The excitation notions used throughout the paper, together with the formal problem statement, are given next.}

 \subsection{Excitation Conditions}

{We start by recalling from~\cite{tao2003adaptive} excitation notions.
\begin{definition}
\label{def:persistency1}
Given $\sigma \geq 0$ and $\mu >0$, a signal $t \mapsto \phi(t) \in \R^n$ is exciting over the finite interval $[\sigma, \sigma+\mu]$ if there exists $\eta>0$ such that
\begin{equation}
\label{eqn:2}
\int_{\sigma}^{\sigma+\mu}\phi(t)\phi^{\top}(t)dt\geq \eta I.
\end{equation}
\end{definition}
\begin{definition}
\label{def:persistency2}
A signal $t \mapsto \phi(t) \in \R^n$ is $\mu$-persistently exciting if there exist $\mu>0$ and $\eta>0$ such that (\ref{eqn:2}) holds for all $\sigma \geq 0$. In addition, a signal $t \mapsto \phi(t) \in \R^n$ is said to be persistently exciting if there exists $\mu>0$ such that $t \mapsto \phi(t) \in \R^n$ is $\mu$-persistently exciting.
\end{definition}
While in Definition~\ref{def:persistency1} the signal $t \mapsto \phi(t)$ is only exciting over the finite interval $[\sigma,\sigma+\mu]$, the $\mu$-persistency of excitation notion in Definition~\ref{def:persistency2} requires the signal  $t \mapsto \phi(t)$ to be exciting over any interval of length $\mu>0$. Persistency of excitation is necessary and sufficient to ensure exponential convergence using the classical gradient descent algorithm~\cite{tao2003adaptive}.}

\subsection{Problem Formulation}
{Using the excitation notions introduced above, the problem addressed in this paper is formulated as follows.}
{
\begin{problem}
\label{prob:1}
Given $\delta>0$ and the linear regression model (\ref{eqn:estimator}) with unknown parameter $\theta^*$, design an estimator such that, if $\phi$ is exciting over $[0,\delta]$ (Definition~\ref{def:persistency1}), then $\theta$ converges to $\theta^*$ in the predetermined finite time $\delta$, and the estimation error stays bounded under bounded measurement noise. The unknown parameter $\theta^*$ is first assumed constant; the piecewise constant and time-varying cases are addressed in Sections~\ref{sec:4} and~\ref{sec:6}, respectively.
\end{problem}
}

{To construct such an estimator and to formally establish its properties, some preliminary tools for hybrid systems~\cite{goebel2009hybrid} are introduced in the next subsection.}

\subsection{Preliminaries on Hybrid Systems}
In this paper, a hybrid system $\mathcal{H}$ has data $(C,f,D,g)$ and is defined by
\begin{equation} \label{eq1}
\begin{split}
\dot{z} & = f(z,u)~~z \in C, \\
 z^+ & =  g(z,u)~~ z\in D,
\end{split}
\end{equation}
where $z \in \R^q$ is the state, $u \in \R^p$ is the input, $f$ is the flow map which captures the continuous dynamics, and $C$ defines the flow set on which $f$ is effective. The
map $g$ defines the jump map and models the discrete
behavior, while $D$ defines the jump set, from which discrete dynamics are allowed. Given an input $u$, a solution to $\mathcal{H}$ is given by the pair $(z,u)$, which is parametrized by $(t,j) \in \R_{\geq0 }\times \N$, where $t$ is the ordinary time keeping track of the flows, and $j$ is the jump index counting the number of jumps (when the system has no input or its input is zero, its solution will be given by $z$). The domain $\dom z= \dom u \subset \R_{\geq 0}\times \N$ of a solution $(z,u)$ to $\mathcal{H}$ is a hybrid time domain, in the sense
that for every $(T,J) \in \dom z$, there exists a nondecreasing sequence $\{t_j\}_{j=0}^{J+1}$ with $t_0=0$ such that $\dom z \cap ([0,T]\times \{0,1,\dots,J\})= \mathop{\bigcup}\limits_{j=0}^J([t_j,t_{j+1}]\times \{j\})$. Given the sequence above, for any $j \in \{1,2,\ldots,J\}$, $t_j$ is the
ordinary time of the $j$-th jump of $z$. A solution to $\mathcal{H}$ is called maximal if it cannot be extended, and it is called complete if its domain is unbounded. For more details on hybrid systems, we refer the interested reader to \cite{sanfelice2021hybrid}. In the remainder of the paper, we focus on hybrid systems without inputs. Hybrid systems with inputs will be considered in Section \ref{sec:5} to analyze the robustness of the proposed hybrid estimator.

\section{Hybrid finite time convergent parameter estimation algorithm}
\label{sec:4}
{This section presents a hybrid parameter estimator for finite-time convergence with respect to constant and piecewise constant unknown parameter $\theta^*$ for the system in (\ref{eqn:estimator}).}

\subsection{Constant Unknown Parameter}
\label{subsection:constant}

The finite time adaptation law to estimate the unknown parameter $\theta^*$ for the linear regression model in (\ref{eqn:estimator}) is formalized as a hybrid system $\mathcal{H}$ and defined as follows:
\begin{equation}
\label{eqn:model2}
\begin{split}
\dot{z} & = f(z)~~z \in C, \\
 z^+ & =  g(z)~~ z\in D,
\end{split}
\end{equation}
with state $z=({\theta}_1,{\theta}_2, \tau_a, \tau_b ,q) \in \mathcal{X}=\R^n\times \R^n\times [0,\delta]\times \R_{\geq 0} \times \N$. The signals $\phi: t \mapsto \phi(t) \in \R^n$ and $y : t \mapsto y(t)\in \R$ are known and used in the hybrid model. The state components $\theta_1$ and $\theta_2$ represent the estimates with update laws based on the adaptation rates $\gamma_1$ and $\gamma_2$, respectively. The state component $\tau_a$ is a timer used to trigger the jump at $\delta>0$. The state components $\tau_b$ and $q$ make it possible to convert the time-varying system into a time invariant one. The flow and jump maps are given by

\small

\begin{equation}
f(z) := \begin{bmatrix}
    -\gamma_1\phi(\tau_b)(\phi^{\top}(\tau_b)\theta_1 - y(\tau_b))\\
    -\gamma_2\phi(\tau_b)(\phi^{\top}(\tau_b)\theta_2 - y(\tau_b))\\
    1\\1\\0
  \end{bmatrix},
g(z) := \begin{bmatrix}
    R({\theta}_1,\theta_2,q)\\
    R({\theta}_1,\theta_2,q)\\
    0\\\tau_b\\1+q
\end{bmatrix},
\label{eqn:fg}
\end{equation}
\normalsize
where $R({\theta}_1,{\theta}_2,q):=K_1(q){\theta}_1+K_2(q){\theta}_2$ and the gains $K_1$ and $K_2$ are given by the functionals
\begin{align}
    K_2(q)&:=I-K_1(q) \nonumber &\text{$\forall$} q \in \N \\ \label{eqn:functionals}
    K_1(q)&:=-\Phi_2(\delta,0) \left(\Phi_1(\delta,0) -\Phi_2(\delta,0)\right)^{-1} &\text{ if } q=0 \\
     K_1(q)&:=I &\text{ if } q\neq 0 \nonumber
\end{align} where $\Phi_i$ is the state-transition matrix of the time-varying system
$\dot{\tilde{\theta}}_i=-\gamma_i\phi(t)\phi^{\top}(t)\Tilde{\theta}_i$.
The flow set is defined as $C:=\mathcal{X}$ and the jump set is given by $D:=\{z \in \mathcal{X}:\tau_a=\delta\}$. The construction of $C$ and $D$ ensures that the system jumps periodically. The logic variable $q$ allows also to ensure that the functional $K_1$ is the identity whenever $q > 0$. For this algorithm, the update based on the term $-\Phi_2(\delta,0) \left(\Phi_1(\delta,0) -\Phi_2(\delta,0)\right)^{-1}$ occurs only one time, when $q=0$. Let us now provide conditions for the functionals $K_1$ and $K_2$ to be well defined. 

\begin{proposition}
\label{pro:invertibility}
Suppose that $\gamma_1\neq \gamma_2$, the regressor $t \mapsto \phi(t)$ is exciting over $[0,\delta]$ (i.e, $\int_{0}^{\delta}\phi(s)\phi^{\top}(s)ds\geq \eta I$ for $\eta>0$), there exists $\phi_M \geq 0$ such that $|\phi(s)| \leq \phi_M$ for all $s\in [0,\delta]$, and the constants $\delta,\gamma_1$ and $\gamma_2$ satisfy the following inequalities:
\begin{equation}
\label{eqn:inequ1}
    0\leq \phi_M^2\gamma_2\delta<1
\end{equation}
\begin{align}
\label{eqn:inequ2}
    0<\left(1-\frac{2\eta\gamma_1}{(1+\phi_M^2\gamma_1\delta)^2}\right)\left(1+\frac{2\gamma_2\phi_M^2\delta}{(1-\phi_M^2\gamma_2\delta)^2}\right) < 1
\end{align}
Then, the gains $K_1$ and $K_2$ in (\ref{eqn:functionals}) are well defined at hybrid time $(\delta,0)$, i.e, $\Phi_1(\delta,0) -\Phi_2(\delta,0)$ is invertible.
\end{proposition}
\begin{proof}
From the invertibility of the state transition matrix $\Phi_2$ and from Lemma~2 in the Appendix, conditions (\ref{eqn:inequ1}) and (\ref{eqn:inequ2}) imply that $|\Phi_1(\delta,0)\Phi_2^{-1}(\delta,0)| \leq |\Phi_1(\delta,0)||\Phi_2^{-1}(\delta,0)| <1$. Hence, by Lemma~3 in the Appendix, the matrix $I-\Phi_1(\delta,0)\Phi_2^{-1}(\delta,0)$ is invertible. Since $\Phi_1(\delta,0)-\Phi_2(\delta,0)=\big(\Phi_1(\delta,0)\Phi_2^{-1}(\delta,0)-I\big)\Phi_2(\delta,0)$ is a product of two invertible matrices, it is invertible, and the functionals $K_1$ and $K_2$ in (\ref{eqn:functionals}) are well defined.
\end{proof}
\begin{remark}
The function defined by $$\delta \mapsto \left(1-\frac{2\eta\gamma_1}{(1+\phi_M^2\gamma_1\delta)^2}\right)\left(1+\frac{2\gamma_2\phi_M^2\delta}{(1-\phi_M^2\gamma_2\delta)^2}\right)$$ evaluated at $\delta=0$ gives $1-2\eta\gamma_1$. Hence, conditions (\ref{eqn:inequ1})-(\ref{eqn:inequ2}) are always satisfied by appropriately choosing the parameters $\delta$, $\gamma_1$ and $\gamma_2$ for given $\eta> 0$ and $\phi_M >0$.
\end{remark}
Next, we show convergence of the parameter estimates ${\theta}_i$, $i \in \{1,2\}$, to $\theta^*$. To this end, we define the following set:
\begin{equation}
    \label{eqn:stableset}
    \mathcal{A}=\{z\in \mathcal{X} : {\theta}_1={\theta}_2=\theta^*\}
\end{equation}

{
Proposition~\ref{pro:invertibility} guarantees that the gains $K_1$ and $K_2$ are well defined at jumps. The following result goes further and provides
explicit bounds on these gains.
\begin{lemma}
\label{lem:Kbound}
Suppose the conditions of Proposition~\ref{pro:invertibility} hold, and let
\begin{equation}
\label{eqn:kappa}
\kappa_1:=\sqrt{1-\frac{2\eta\gamma_1}{(1+\phi_M^2\gamma_1\delta)^2}},
\qquad
\kappa_2:=\sqrt{1+\frac{2\gamma_2\phi_M^2\delta}{(1-\phi_M^2\gamma_2\delta)^2}}.
\end{equation}
Then, the gains $K_1$ and $K_2$ in (\ref{eqn:functionals}) satisfy
\begin{equation}
\label{eqn:Kbound}
|K_1|\le\frac{1}{1-\kappa_1\kappa_2},
\qquad
|K_2|=|I-K_1|\le\frac{\kappa_1\kappa_2}{1-\kappa_1\kappa_2},
\end{equation}
and, moreover,
\begin{equation}
\label{eqn:diffbound}
|(\Phi_1(\delta,0)-\Phi_2(\delta,0))^{-1}|\le\frac{\kappa_2}{1-\kappa_1\kappa_2}.
\end{equation}
\end{lemma}
}
{
\begin{proof}
We have that $K_1=-\Phi_2(\delta,0)(\Phi_1(\delta,0)-\Phi_2(\delta,0))^{-1}=(I-\Phi_1(\delta,0)\Phi_2^{-1}(\delta,0))^{-1}$. Hence, since $|\Phi_1\Phi_2^{-1}|\le\kappa_1\kappa_2<1$ from Proposition \ref{pro:invertibility}, it follows from Lemma $2$ in the Appendix that $|K_1| \leq \frac{1}{1-\kappa_1\kappa_2}$ and $|I-K_1| \leq \frac{\kappa_1\kappa_2}{1-\kappa_1\kappa_2}$. Finally, from $K_1=-\Phi_2(\delta,0)(\Phi_1(\delta,0)-\Phi_2(\delta,0))^{-1}$, we have $(\Phi_1(\delta,0)-\Phi_2(\delta,0))^{-1}=-\Phi_2^{-1}(\delta,0)K_1$, which implies (\ref{eqn:diffbound}).\end{proof}
}
\begin{theorem}
\label{prop:3}
Let $\delta>0$, $\gamma_1, \gamma_2 \geq0$ and a regressor $t \mapsto \phi(t)\in \R^n$ be given. Consider the hybrid system $\mathcal{H}$ in (\ref{eqn:model2}). Assume the unknown parameter $\theta^*\in \R^n$ is constant. If the conditions in Proposition~\ref{pro:invertibility} hold, then there exists $\alpha \in \mathcal{K}_{\infty}$ such that for each maximal solution $z$ to $\mathcal{H}$ from $z(0,0) \in \mathcal{X}_0:=\{z \in \mathcal{X}: \theta_1=\theta_2, \tau_a=\tau_b=q=0\}$, we have that $|z(t,j)|_{\mathcal{A}}\leq \alpha( |z(0,0)|_{\mathcal{A}})$ for all $(t,j)\in \dom z$, and ${\theta}_1(t,j)={\theta}_2(t,j)=\theta^*$ for all $(t,j) \in \dom z$ satisfying $t \geq \delta$ and $j\geq 1$.
\end{theorem}

\begin{proof} First let us show the finite time convergence of $\theta_i$ to $\theta^*$, $i\in \{1,2\}$. To do so, we will analyse the evolution of the parameters errors along a solution. From $\mathcal{X}_0$, the evolution of the parameters error before the first jump, i.e, $t \in [0,\delta]$ and $j=0$, is given by
\begin{equation}
\label{eqn:solu22}
\tilde{\theta}_i(t,0)=\Phi_i(t,0)\tilde{\theta}_i(0,0) \text{ for all } i\in\{1,2\}
\end{equation}
Then, for the solution to be maximal, a jump occurs at $(t,j)=(\delta,0)$. After the first jump and from the expressions of $K_1$ and $K_2$ in (\ref{eqn:functionals}), which are well defined at jumps according to Proposition~\ref{pro:invertibility}, we have
    \begin{align}
    \Tilde{\theta}_i(\delta,1)&=
     R({\theta}_1(\delta,0),{\theta}_2(\delta,0),q(\delta,0))-\theta^* \nonumber \\
    &=K_1(q(\delta,0)) \left( \Phi_1(\delta,0)-\Phi_2(\delta,0) \right) \Tilde{\theta}_2(0,0)+\Tilde{\theta}_2(\delta,0)    \nonumber   \\  \label{eqn:pro}
    &=\Phi_2(\delta,0)(-\Tilde{\theta}_2(0,0)+\Tilde{\theta}_2(0,0))=0.
\end{align}
The second equality comes from (\ref{eqn:solu22}) and the fact that $\Tilde{\theta}_1(0,0)=\Tilde{\theta}_2(0,0)$.
Then, we have $\Tilde{\theta}_i(\delta,1)=0$, which implies that ${\theta}_i(\delta,1)=\theta^*$. Hence, we have the finite time convergence of $\theta_i$, $i\in \{1,2\}$, to $\theta^*$.

To show the boundedness of the solutions, we consider the following Lyapunov function:
$V(z)=\tilde{\theta}^{\top}_1P_1\tilde{\theta}_1+\tilde{\theta}^{\top}_2P_2\tilde{\theta}_2$, where $P_1$ and $P_2$ are symmetric positive definite matrices. It can be shown that $V$ satisfies the following property:
$$\alpha_1(|z|_{\mathcal{A}})\leq V(z) \leq \alpha_2(|z|_{\mathcal{A}}) ~~~~~ \forall z \in C\cup D\cup g(D),$$
where $\alpha_1(r)=\underline{\alpha}_1r^2$ with $\underline{\alpha}_1=\min(\underline{\mu}(P_1),\underline{\mu}(P_2))$
and $\alpha_2(r)=\overline{\alpha}_2r^2$ with $\overline{\alpha}_2=\max(\overline{\mu}(P_1),\overline{\mu}(P_2))$.
Now for each $z\in C$, we have $\langle\nabla V(z),f(z)\rangle=\begin{bmatrix}\tilde\theta_1^{\top} & \tilde\theta_2^{\top}\end{bmatrix} Q(\tau_b)\begin{bmatrix}\tilde\theta_1^{\top} & \tilde\theta_2^{\top}\end{bmatrix}^{\top}$
with{\small $$Q(\tau_b)=\begin{bmatrix}
-\gamma_1\He(\phi(\tau_b)\phi^{\top}(\tau_b),P_1) & 0 \\
0 & -\gamma_2\He(\phi(\tau_b)\phi^{\top}(\tau_b),P_2)
\end{bmatrix}$$} The matrix $Q$ is negative semidefinite for all $\tau_b \in \R_{\geq 0}$. Hence, $\langle\nabla V(z),f(z)\rangle \leq 0$ for all $z \in C$.
For the variations of $V$ at jump instants, let $z$ be a maximal solution for the hybrid system $\mathcal{H}$ with $z(0,0) \in \mathcal{X}_0$, and let $(t,j)$ be such that $(t,j+1) \in \dom z$. If $j=0$, it follows directly from (\ref{eqn:pro}) that
$V(z(t,j+1))-V(z(t,j))=-V(z(t,j)) \leq -\alpha_1(|z(t,j)|_{\mathcal{A}})\leq 0$. Similarly, if $j\in \mathbb{N}_{\geq 1}$, it follows from (\ref{eqn:functionals}) that $V(z(t,j+1))-V(z(t,j)) \leq -\alpha_1(|z(t,j)|_{\mathcal{A}})\leq 0$. Hence, we have that $\alpha_1(|z(t,j)|_{\mathcal{A}})\leq V(z(t,j)) \leq V(z(0,0)) \leq \alpha_2(|z(0,0)|_{\mathcal{A}})$, which implies that $|z(t,j)|_{\mathcal{A}}\leq \alpha_1^{-1}\circ \alpha_2( |z(0,0)|_{\mathcal{A}})$ for all $(t,j)\in \dom z$. By defining $\alpha=\alpha_1^{-1}\circ \alpha_2$ which belongs to the class $\mathcal{K}_{\infty}$, it follows that $|z(t,j)|_{\mathcal{A}}\leq \alpha( |z(0,0)|_{\mathcal{A}})$ for all $(t,j)\in \dom z$.  
\end{proof}

\begin{remark}
    {Let us clarify the intuition behind the choice of the functionals $K_1$ and $K_2$. At hybrid time $(\delta, 0)$, we have:
\begin{align*}
\tilde{\theta}_i(\delta, 0) &= \Phi_i(\delta, 0) \tilde{\theta}_i(0,0), i=1,2
\end{align*}
Since $\tilde{\theta}_1(0,0) = \tilde{\theta}_2(0,0) = \tilde{\theta}_0$ (from the initial condition set $\mathcal{X}_0$), we want to find $K_1, K_2$ such that $\tilde{\theta}_i(\delta, 1) = K_1\theta_1(\delta, 0) + K_2\theta_2(\delta, 0) - \theta^* = 0$, with $K_2 = I - K_1$, this implies that $K_1(\Phi_1 - \Phi_2)\tilde{\theta}_0 + \Phi_2\tilde{\theta}_0 = 0$
For this to hold for all possible initializations $\tilde{\theta}_0$, we choose $K_1 = -\Phi_2(\Phi_1 - \Phi_2)^{-1}$.}
\end{remark}

\begin{remark}
The choice of initial conditions of the estimate is not critical. Indeed, we can choose different initial values $\theta_1(0,0) \neq \theta_2(0,0)$. In this case, one can only ensure the convergence of the estimates $\theta_1$ and $\theta_2$ to $\theta^*$ at the second jump, i.e, $(t,j)=(2\delta,2)$. In this case, the new expressions of the gains $K_1$ and $K_2$ are given by
\begin{align*}
K_2(q)&=I-K_1(q) \quad \forall q \in \N\\
K_1(q)&=-\Phi_2(2\delta,\delta)\left(\Phi_1(2\delta,\delta)-\Phi_2(2\delta,\delta)\right)^{-1} \text{ if } q=1\\
K_1(q)&=I \text{ if } q\neq 1
\end{align*}
which are well defined at each jump of the hybrid system $\mathcal{H}$ in (\ref{eqn:model2}), if $t\mapsto\phi(t)$ is exciting and bounded over $[\delta,2\delta]$ and conditions (\ref{eqn:inequ1})-(\ref{eqn:inequ2}) are satisfied.
\end{remark}

\begin{remark}
In general, the initial condition on the timer ($\tau_a(0,0)=0$) can be removed at the cost of additional requirements on the regressor $\phi$. Indeed, if the initial condition on the timer can be freely chosen, i.e, $\tau_a(0,0) \in [0,\delta]$, and in view of Proposition~\ref{pro:invertibility}, the gains $K_1$ and $K_2$ are well defined if there exists an interval $[0,\delta-s]$ for some $s \in [0,\delta)$ such that the regressor $t\mapsto\phi(t)$ is exciting and bounded and conditions (\ref{eqn:inequ1})-(\ref{eqn:inequ2}) are satisfied. Let us mention that this new excitation requirement on the regressor is relatively close to condition (10) in~\cite{wang2019robust}, allowing to ensure finite time convergence of the parameter error. In this case, the new expressions of the gains $K_1$ and $K_2$ will be given by
\begin{align*}
K_2(q)&=I-K_1(q) \quad \forall q \in \N\\
K_1(q)&=-\Phi_2(\delta-\tau_a(0,0),0)\big(\Phi_1(\delta-\tau_a(0,0),0)\\
&\quad -\Phi_2(\delta-\tau_a(0,0),0)\big)^{-1} \text{ if } q=0\\
K_1(q)&=I \text{ if } q\neq 0.
\end{align*}
\end{remark}

\subsection{Piecewise Constant Unknown Parameter}

When the unknown parameter $\theta^*$ is a piecewise constant function, it is also possible to estimate it in finite time. However, one jump is not enough. Therefore, recursive jumps are embedded in the parameter estimator. For this purpose, we consider the hybrid finite time convergent adaptation law in (\ref{eqn:model2}), with state $z=({\theta}_1,{\theta}_2, \tau_a,\tau_b,q) \in \mathcal{X}=\R^n\times \R^n\times [0,\delta]  \times \R_{\geq 0}  \times \N$, with data $f$, $g$, $C$, $D$ as in (\ref{eqn:model2}), but with 

\small
\begin{align}
\label{eqn:functionalss}
    K_1(q)&=-\Phi_2((q+1)\delta,q\delta)\left(\Phi_1((q+1)\delta,q\delta) -\Phi_2((q+1)\delta,q\delta)\right)^{-1} \nonumber\\
    K_2(q)&=I-K_1(q)
\end{align}

\normalsize
As in the previous section, for $K_1$ and $K_2$ to be well defined at jumps, we have the following result for which the proof follows the same lines as the proof of Proposition~\ref{pro:invertibility}.
\begin{proposition}
\label{pro:invertibility2}
Suppose that $\gamma_1\neq \gamma_2$, there exists $\mu>0$ such that $\mu\leq \delta$ and the regressor $t \mapsto \phi(t)$ is $\mu$-persistently exciting, there exists $\phi_M \geq 0$ such that $|\phi(s)| \leq \phi_M$ for all $s \in \R_{\geq 0}$ and that conditions (\ref{eqn:inequ1})-(\ref{eqn:inequ2}) are satisfied. 
Then, the gains $K_1$ and $K_2$ in (\ref{eqn:functionalss}) are well defined at jumps.
\end{proposition}
\begin{theorem}
\label{thm:main}
Let $\delta>0$, $\gamma_1, \gamma_2 \geq0$, and consider the hybrid system $\mathcal{H}$ in (\ref{eqn:model2}), with $K_1$ and $K_2$ defined in (\ref{eqn:functionalss}). Assume the unknown parameter $\theta^*:[0,+\infty) \rightarrow \R^n$ is piecewise constant, where the time instants at which the parameter changes values are defined by a sequence $\{d_k\}_{k\in \N}$ satisfying  $\cup_{k=0}^{+\infty}[d_k,d_{k+1})=[0,+\infty)$. If the conditions in Proposition~\ref{pro:invertibility2} hold and if the parameter $\delta$ is chosen such that
\begin{equation}
\label{eqn:cdt_delta}
    0<2\mu\leq 2\delta <\min_{k\in \N}\{d_{k+1}-d_k\}
\end{equation}
then for each maximal solution $z$ to $\mathcal{H}$ from  $z(0,0)\in \mathcal{X}_0:=\{z \in \mathcal{X} : \theta_1=\theta_2, \tau_a=\tau_b=q=0\}$, the following property is satisfied: for each $j \in \N_{\geq 1}$ there exists an interval with nonempty interior $I'_j\subseteq I_j \cup I_{j+1}$ such that ${\theta}_1(t,j)={\theta}_2(t,j)=\theta^*$ for all $t \in I'_j$.
\end{theorem}

\begin{proof}
Let $z$ be the maximal solution of the hybrid system $\mathcal{H}$ in (\ref{eqn:model2}). We have from (\ref{eqn:cdt_delta}) that the parameter $\theta^*$ is constant over the first two intervals $I_0$ and $I_1$ and using the fact that persistence of excitation implies interval excitation we have that $\phi$ is exciting on $I_0$. Hence, we have the existence of an interval with non empty interior $I'_0 \subseteq I_1 \subseteq I_0 \cup I_1$ such that ${\theta}_1(t,j)={\theta}_2(t,j)=\theta^*$ for all $t \in I'_1$.

Let $j\geq 1$ and let us consider the two consecutive intervals $I_j$ and $I_{j+1}$. From definition of the jump map $g$ it follows directly that ${\theta}_1(j\delta,j)={\theta}_2(j\delta,j)$ for all $j\geq 1$ and one of the following three scenarios hold:
\begin{itemize}
    \item if ${\theta}_1(j\delta,j)={\theta}_2(j\delta,j)=\theta^*$ and the parameter $\theta^*$ does not change its value at $t=j\delta$, then one can ensure the existence of an interval $I'_j \subseteq I_j=[j\delta, \delta (j+1)]$ with non empty interior such that ${\theta}_1(t,j)={\theta}_2(t,j)=\theta^*$ for all $t \in I'_j$;
    
    \item if ${\theta}_1(j\delta,j)={\theta}_2(j\delta,j)=\theta^*$ and the parameter $\theta^*$ change its value at $t=j\delta$. Then from (\ref{eqn:cdt_delta}) it follows that $\theta^*$ will keep constant on the interval $I_j\cup I_{j+1}$. Moreover, from (\ref{eqn:cdt_delta}) we have that $\phi$ is exciting on $I_j$, Hence, from Proposition~\ref{prop:3} one can ensure that ${\theta}_1(t,j+1)={\theta}_2(t,j+1)=\theta^*$ for all $t \in I_{j+1}$;
    
    \item if ${\theta}_1(j\delta,j)={\theta}_2(j\delta,j)\neq \theta^*$, then the parameter $\theta^*$ changes its value on the interval $I_{j-1}$ and from (\ref{eqn:cdt_delta}) it will keep constant on the interval $I_j$ and at least on an non empty sub-interval $I'_{j}$ of $I_{j+1}$. Moreover, we have from (\ref{eqn:cdt_delta}) that $\phi$ is exciting on $I_j$, it follows then from Proposition~\ref{prop:3} that ${\theta}_1(t,j+1)={\theta}_2(t,j+1)=\theta^*$ for all $t \in I'_j$.
\end{itemize}
Hence, based on the three separate cases examined above, the result follows directly.
\end{proof}
The previous result shows that the parameter $\theta^*$ is exactly estimated, after a finite time since the time it changed. Indeed, the intervals $I'_j$ imply that whenever the parameter changes its value, the proposed parameter estimator converges to the exact value $\theta^*$ no later than $2\delta$ seconds after the parameter value changed. Let us also mention that similarly to Theorem~\ref{prop:3}, the same bounds on the solutions can be established.

\section{Robustness to measurement noise}
\label{sec:5}
In this section, we analyse the robustness of the proposed hybrid parameter estimator with respect to bounded time-varying measurement noise. For the sake of readability, we focus on constant unknown parameters. However, the robustness results can be generalized using the same approach to deal with piecewise constant unknown parameters.

The linear regression model with measurement noise is expressed as
\begin{equation}
 \label{eqn:estimator_measurements}
     y(t)=\theta^{*{\top}}\phi(t)+w(t)
 \end{equation}
where $t \mapsto w(t)\in \R$ represents the additive noise in the measurements of $y$. 
 The estimator $\hat{y}$ of the actual output $y$ is defined as $\hat{y}(t):=\hat{\theta}^{\top}(t)\phi(t)$ and the error between the actual and estimated outputs is given by $e(t):=\hat{y}(t)-y(t)=\Tilde{\theta}^{\top}(t)\phi(t) + w(t)$. The parameter estimation error for the classical continuous-time gradient algorithm\cite{narendra2012stable} is given by $
     \dot{\Tilde{\theta}}(t)=-\gamma\phi(t)\phi^{\top}(t)\Tilde{\theta}(t)+\gamma\phi(t)w(t)$, where $\gamma>0$ is the adaptation rate. Starting from the noise-free hybrid estimator defined in (\ref{eqn:model2}), our hybrid parameter estimator under the effect of the measurement noise $w$ is defined as a hybrid system $\mathcal{H}_{\nu}$ with data $(C,f_{\nu},D,g_{\nu})$ and described as follows:
\begin{equation}
\label{eqn:model3}
\begin{split}
\dot{z}_{\nu} & = f_{\nu}(z_{\nu})~~z_{\nu} \in C, \\
 z_{\nu}^+ & =  g_{\nu}(z_{\nu})~~ z_{\nu}\in D,
\end{split}
\end{equation}
with $\nu(t) =
\left[
\begin{array}{ccccc}
\gamma_1 \phi(t) w(t) &
\gamma_2 \phi(t) w(t) &
1 &
1 &
0
\end{array}
\right]^{\top}$, and the flow and jump maps are given by $f_{\nu}=f+\nu$ and $g_{\nu}=g$ with $f$ and $g$ in (\ref{eqn:fg}). The state $z_{\nu}=({\theta}_1,{\theta}_2, \tau_a, \tau_b ,q) \in \mathcal{X}=\R^n\times \R^n\times [0,\delta]\times \R_{\geq 0} \times \N$ and the flow and jump sets are given by $C:=\mathcal{X}$ and $D:=\{z \in \mathcal{X}:\tau_a=\delta\}$. 

 To analyze the effect of measurement noise, we rely on the robustness tools developed in the hybrid systems framework~\cite{goebel2009hybrid}. Consider the noise-free hybrid system $\mathcal{H}$ in (\ref{eqn:model2}) and the noisy hybrid system $\mathcal{H}_{\nu}$ in (\ref{eqn:model3}). We have the following result, showing closeness of the trajectories of the noisy hybrid system $\mathcal{H}_{\nu}$ and the noise-free hybrid system $\mathcal{H}$. We rely on the notion of $(\tau,\varepsilon)$-closeness of trajectories (See Definition $4.11$ in \cite{goebel2009hybrid}).
\begin{proposition}
\label{pro:robust1}
Consider $\delta>0$, $ \gamma_1, \gamma_2 >0$ and a regressor $t \mapsto \phi(t)\in \R^n$. Consider the noise-free hybrid system $\mathcal{H}$ in (\ref{eqn:model2}) and the noisy hybrid system $\mathcal{H}_{\nu}$ in (\ref{eqn:model3}). Let $\mathcal{K} \subseteq \R^{2n}$ be a compact set and let $\tau,\varepsilon >0$. If the conditions in Proposition~\ref{pro:invertibility} hold, then there exists $\Bar{\nu}>0$ such that if $|\nu(t)|\leq \Bar{\nu}$ for all $t \in \R_{\geq 0}$, the following holds: for every maximal solution $z_{\nu}$ to $\mathcal{H}_{\nu}$ with an initial condition $z_{\nu}(0,0) \in \mathcal{K}\times [0,\delta]\times \R_{\geq 0} \times \N \cap \mathcal{X}_0$, with $\mathcal{X}_0=\{z \in \mathcal{X}: \theta_1=\theta_2, \tau_a=\tau_b=q=0\}$, there exists a solution $z$ to $\mathcal{H}$ with initial condition $z(0,0) \in \mathcal{K}\times [0,\delta] \times \R_{\geq 0} \times \N \cap \mathcal{X}_0$ such that $z_{\nu}$ and $z$ are $(\tau, \varepsilon)$ close.
\end{proposition}

\begin{proof}
To prove the robustness result, we rely on Proposition $6.34$ in~\cite{goebel2009hybrid}. Indeed, the conditions given in Proposition $6.34$ in~\cite{goebel2009hybrid} are the well posedness (see Definition $6.27$ in~\cite{goebel2009hybrid}) of the hybrid system $\mathcal{H}$ and the pre-forward completeness (see Definition $6.12$ in~\cite{goebel2009hybrid}) of its maximal trajectories. First, from the construction of the map $f_{\nu}$ and due to the boundedness of $\phi$ one can easily check that the hybrid system $\mathcal{H}$ satisfies the hybrid basic conditions (see Assumption $6.5$ in~\cite{goebel2009hybrid}). Hence, we have from Theorem $6.30$ in~\cite{goebel2009hybrid} that the system $\mathcal{H}$ is well-posed. We show that every maximal solution to the system $\mathcal{H}$ is complete. First, the solution $z$ to $\mathcal{H}$ does not escape to infinity in finite time. Let $T_C(z)$ be the tangent cone of the flow set $C$ at $z$. We have for each $z(0,0) \in C \setminus D$, $T_C(z(0,0))\cap f(z(0,0))\neq \emptyset$. Moreover, we have that $g(D) \subset C \subset C\cup D$. Hence, item (a) of Proposition $6.10$ in~\cite{goebel2009hybrid} holds, which implies that every maximal solution to the system $\mathcal{H}$ is complete. The closeness result follows then from Proposition $6.34$ in~\cite{goebel2009hybrid}.
\end{proof}
Next, we show that under the PE condition in Definition \ref{def:persistency2}, the noisy system $\mathcal{H}_{\nu}$ is input to state stable (ISS) for any essentially bounded measurement noise $w$.
\begin{proposition}
\label{pro:robust2}
Let $\delta>0$, $\gamma_1, \gamma_2 \geq0$ and a regressor $t \mapsto \phi(t)\in \R^n$. Consider the noisy hybrid system $\mathcal{H}_{\nu}$ in (\ref{eqn:model3}). If the conditions in Proposition~\ref{pro:invertibility} hold, the regressor $\phi(t)$ is $\delta$-persistently exciting, there exists $\phi_M > 0$ such that $|\phi(t)| \leq \phi_M$ for all $t \geq 0$, then each maximal solution $z_{\nu}$ to $\mathcal{H}_{\nu}$ from $z_{\nu}(0,0) \in \mathcal{X}_0=\{z_{\nu} \in \mathcal{X}: \theta_1=\theta_2, \tau_a=\tau_b=q=0\}$ satisfies
\begin{eqnarray}
   |z_{\nu}(t,j)|_{\mathcal{A}}&\leq \rho(j)\beta(|z_{\nu}(0,0)|_{\mathcal{A}},t)+ \alpha_1(|w|_{\infty}) \nonumber \\ \label{eqn:robustaa}  &+(1-\rho(j))(\alpha_2(|w|_{\infty})) 
\end{eqnarray}
with $\rho(0)=1$ and $\rho(j)=0$ for $j \in \N_{>0}$, $\beta(s,t)=\exp(\delta)\exp(-t)s$, $\alpha_1(s)=\max(\gamma_1,\gamma_2)\phi_M\delta s$ and
$\alpha_2(s)=\phi_M(\frac{1}{1-\kappa_1\kappa_2}\delta(\gamma_1+\gamma_2)+\gamma_1(\frac{\kappa}{\lambda}+\delta)+\gamma_2 \delta) s$, where $\kappa_1, \kappa_2$ are given in (\ref{eqn:kappa}) and $\kappa,\lambda$ are the constants of Lemma~2 in the Appendix.
\end{proposition}
\begin{proof}
Let $z_{\nu}$ be the maximal solution of the hybrid system $\mathcal{H}_{\nu}$ in (\ref{eqn:model3}), from an initial condition $z_{\nu}(0,0)\in \mathcal{X}_0$. Let $\Tilde{\theta}_i=\theta_i-\theta^*$ for each $i\in \{1,2\}$. Let us analyze the error dynamics. The evolution of the parameters error before the first jump, i.e, $t \in [0,\delta]$ and $j=0$, is given by $\Tilde{\theta}_i(t,0)=\Phi_i(t,0)\Tilde{\theta}_i(0,0)+\int_0^t\gamma_i\Phi_i(s,0)\phi(s)w(s)ds$. Hence, from Lemma~$2$ in the Appendix and using the fact that $t \mapsto \phi(t)$ is $\delta$-persistently exciting, we have
    \begin{align}
     |\Tilde{\theta}_i(t,0)|&\leq  |\Phi_i(t,0)\Tilde{\theta}_i(0,0)|+\gamma_i\phi_M|w|_{\infty}\int\limits_0^t|\Phi_i(s,0)|ds \nonumber
      \\[-8pt] \label{eqn:robusta}
     &\leq \beta(|\Tilde{\theta}_i(0,0)|,t)+\gamma_i\phi_M|w|_{\infty}\delta
 \end{align}
 where the last inequality follows from the fact that the first jump does not occur before $\delta$ and that $|\Phi_i(t,0)| \leq 1$ for all $t \geq 0$. Then, for the solution to be maximal, a jump occurs at $(t,j)=(\delta,0)$. After the first jump and from the expressions of $K_1$ and $K_2$, which are well defined according to Proposition~\ref{pro:invertibility}. Moreover, using the fact that $z_{\nu}(0,0) \in \mathcal{X}_0$ and the description of the functional $K_1$ in (\ref{eqn:functionals}), we have that,
 
\begin{align}
\label{eqn:prop6}
|\Tilde{\theta}_i(\delta,1)|  =\bigg|-K_1(q(\delta,0))\int_0^{\delta}\gamma_1\Phi_1(s,0)\phi(s)w(s)ds  +  (K_1(q(\delta,0))-I)\int_0^{\delta}\gamma_2\Phi_2(s,0)\phi(s)w(s)ds\bigg|.
\end{align}

Then, one gets from (\ref{eqn:prop6}) and Lemma \ref{lem:Kbound} that
 \begin{equation}
 \label{eqn:prop6_1}
     |\Tilde{\theta}_i(\delta,1)| \leq \frac{1}{1-\kappa_1\kappa_2}\delta\phi_M|w|_{\infty}(\gamma_1+\gamma_2).
 \end{equation}
Since the first component $\tilde{\theta}_1$ has only one nontrivial jump at $t=\delta$, $\tilde{\theta}_1$ at $(t,j)$ for $t \geq \delta$ and $j \geq 1$ is given by $\Tilde{\theta}_1(t,j)=\Phi_1(t,\delta)\Tilde{\theta}_1(\delta,1)+\gamma_1\int\limits_{\delta}^t\Phi_1(s,\delta)\phi(s)w(s)ds$. Hence, using the fact that $ t \mapsto \phi(t)$ is $\delta$-persistently exciting, we have for $j\delta \leq t < (j+1) \delta$, $j \in \mathbb{N}_{\geq 1}$, that
\begin{align}
 |\tilde{\theta}_1&(t,j)| \leq\frac{1}{1-\kappa_1\kappa_2}\delta\phi_M|w|_{\infty}(\gamma_1+\gamma_2) \nonumber \\ &+\gamma_1\phi_M|w|_{\infty}\left(\int\limits_{\delta}^{j\delta}\kappa e^{-\lambda (s-\delta)}ds+\delta \right) \nonumber \\
 \label{eqn:robustc} 
 &\leq \phi_M|w|_{\infty}\left(\frac{1}{1-\kappa_1\kappa_2}\delta(\gamma_1+\gamma_2)+\gamma_1\left(\frac{\kappa}{\lambda}+\delta\right)\right)
\end{align}
where the first inequality comes from Lemma~$2$ in the Appendix, the fact that $|\Phi_1(u,\delta)| \leq 1$ for all $u \in [j\delta,t]$ and the use of the triangle inequality. The second component $\tilde{\theta}_2$ has successive jumps at $t=j\delta$, $j\geq 1$, the solution of $\tilde{\theta}_2$ at $(t,j)$ for $ j\delta \leq t < (j+1) \delta $ and $j \geq 1$ is given by $\Tilde{\theta}_2(t,j)=\Phi_2(t,j\delta)\Tilde{\theta}_1(j\delta,j)+\gamma_2\int\limits_{j\delta}^t\Phi_2(s,j\delta)\phi(s)w(s)ds$.
Hence, using the fact that $ t \mapsto \phi(t)$ is $\delta$-persistently exciting, we have for $j \delta \leq t  < (j+1) \delta$, $j \in \mathbb{N}_{\geq 1}$, that

\begin{align}
  |\Tilde{\theta}_2(t,j)| &\leq 
 |\Tilde{\theta}_1(j\delta,j)| +\gamma_2\phi_M|w|_{\infty}\left|\int\limits_{j\delta}^t\Phi_2(s,j\delta)ds \right| \nonumber\\ 
 &\leq  |\Tilde{\theta}_1(j\delta,j)|+ \gamma_2\phi_M|w|_{\infty}\delta \nonumber \\ 
 &\leq \phi_M|w|_{\infty}\left(\frac{1}{1-\kappa_1\kappa_2}\delta(\gamma_1+\gamma_2)+\gamma_1\left(\frac{\kappa}{\lambda}+\delta \right)\right) \nonumber \\ \label{eqn:robustcc}
 &\quad + \phi_M|w|_{\infty}\gamma_2 \delta
 \end{align}

 where the first and second inequalities follow from the fact that $|\Phi_2(u,j\delta)| \leq 1$ for all $u \in [j\delta,t]$ and the last inequality comes from (\ref{eqn:robustc}). Therefore, we can unify the bounds obtained in (\ref{eqn:robusta}), (\ref{eqn:prop6_1}), (\ref{eqn:robustc}) and (\ref{eqn:robustcc}) for all $(t,j) \in \dom z_{\nu}$ to get (\ref{eqn:robustaa}).
\end{proof}
\begin{remark}
Intuitively, the bounds in (\ref{eqn:robustaa}) capture the effect of the initial condition $z_{\nu}(0,0)$ and the measurement noise $w$ on the parameter estimation error. It can be seen from (\ref{eqn:prop6_1}) that while the hybrid estimator can estimate the unknown parameter at the first jump $(t,j)=(\delta,1)$ in the noise-free case, the estimation error at the first jump in the noisy case is bounded by $\frac{1}{1-\kappa_1\kappa_2}\delta\phi_M|w|_{\infty}(\gamma_1+\gamma_2)$. Hence, the estimation error in the noisy case is smaller when $\kappa_1\kappa_2$ is smaller. Then, to make the estimation error smaller at the jump, the idea is to select the estimation parameters $\delta$, $\gamma_1$ and $\gamma_2$ to make the term $\kappa_1\kappa_2$ closer to $0$ while ensuring the satisfaction of conditions (\ref{eqn:inequ1})-(\ref{eqn:inequ2}).
\end{remark}

We now show that under finite-time excitation condition, the noisy system $\mathcal{H}_{\nu}$ is integral input to state stable (iISS) \cite{angeli2000characterization} with respect to measurement noise $w$.

\begin{proposition}
\label{pro:robust3}
Let $\delta>0$, $\gamma_1, \gamma_2 \geq0$ and a regressor $t \mapsto \phi(t)\in \R^n$. Consider the noisy hybrid system $\mathcal{H}_{\nu}$ in (\ref{eqn:model3}). If the conditions in Proposition~\ref{pro:invertibility} hold and there exists $\phi_M > 0$ such that $|\phi(t)| \leq \phi_M$ for all $t \geq0$, then any solution $z_{\nu}$ to $\mathcal{H}_{\nu}$ from $z_{\nu}(0,0) \in \mathcal{X}_0=\{z_{\nu} \in \mathcal{X}: \theta_1=\theta_2, \tau_a=\tau_b=q=0\}$ satisfies
\begin{align}
   |z_{\nu}(t,j)|_{\mathcal{A}}&\leq \rho(j)\beta(|z_{\nu}(0,0)|_{\mathcal{A}},t)+ \int_0^t\alpha_1(|w(s)|)ds \nonumber \\ \label{eqn:robustaa2}  &+(1-\rho(j))\left(\int_0^t\alpha_2(|w(s)|)ds\right) 
\end{align}
with $\rho(0)=1$ and $\rho(j)=0$ for $j \in \N_{>0}$, $\beta(s,t)=\exp(\delta)\exp(-t)s$, $\alpha_1(s)=\max\{\gamma_1,\gamma_2\}\phi_Ms$ and
$\alpha_2(s)=(\frac{1}{1-\kappa_1\kappa_2}\delta+1) (\gamma_1+\gamma_2)\phi_Ms$, where $\kappa_1, \kappa_2$ are given in (\ref{eqn:kappa}) and $\kappa,\lambda$ are the constants of Lemma~2 in the Appendix.
\end{proposition}
\begin{proof}
Let $z_{\nu}$ be the maximal solution of the hybrid system $\mathcal{H}_{\nu}$ in (\ref{eqn:model3}), from an initial condition $z_{\nu}(0,0)\in \mathcal{X}_0$. Let $\Tilde{\theta}_i=\theta_i-\theta^*$ for each $i\in \{1,2\}$, and let us analyze the error dynamics. For each $i$, the evolution of the parameters error $\tilde{\theta}_i$ before the first jump, i.e, $t \in [0,\delta]$ and $j=0$, is given by $\Tilde{\theta}_i(t,0)=\Phi_i(t,0)\Tilde{\theta}_i(0,0)+\int_0^t\gamma_i\Phi_i(s,0)\phi(s)w(s)ds$. Hence, from Lemma $3$ in the Appendix and using the fact that $t \mapsto \phi(t)$ is exciting over the interval $[0, \delta]$, we have
    \begin{align}
     |\Tilde{\theta}_i(t,0)| \leq \beta(|\Tilde{\theta}_i(0,0)|,t)+\gamma_i\phi_M\int\limits_0^t|w(s)|ds
     \label{eqn:robusta2}
 \end{align}
 where the inequality follows from the fact that the first jump does not occur before $\delta$ and that $|\Phi_i(t,0)| \leq 1$ for all $t \geq 0$. 
 
 Then, for the solution to be maximal, a jump occurs at $(t,j)=(\delta,0)$. After the first jump and from the expressions of $K_1$ and $K_2$, which are well defined according to Proposition~\ref{pro:invertibility}, we have the chain of equalities in (\ref{eqn:prop6}). Then, one gets from (\ref{eqn:prop6}) and Lemma \ref{lem:Kbound} that
 \begin{equation}
 \label{eqn:prop6_2}
     |\Tilde{\theta}_i(\delta,1)| \leq \frac{1}{1-\kappa_1\kappa_2}\delta\phi_M(\gamma_1+\gamma_2)\int\limits_0^t|w(s)|ds.
 \end{equation}
 Since the first component $\tilde{\theta}_1$ has only one nontrivial jump at $t=\delta$, the solution of $\tilde{\theta}_1$ at $(t,j)$ for $t \geq \delta$ and $j \geq 1$ is given by
$\Tilde{\theta}_1(t,j)=\Phi_1(t,\delta)\Tilde{\theta}_1(\delta,1)+\gamma_1\int\limits_{\delta}^t\Phi_1(s,\delta)\phi(s)w(s)ds$. Hence, from (\ref{eqn:prop6_2}), we have for all $t\geq \delta$ and $j\geq 1$,
\begin{align}
|\Tilde{\theta}_1&(t,j)| \leq \frac{1}{1-\kappa_1\kappa_2}\delta\phi_M(\gamma_1+\gamma_2)\int\limits_0^t|w(s)|ds \nonumber \\
&+\gamma_1\phi_M\int\limits_{\delta}^{t}|w(s)|ds  \nonumber \\
&\leq \left(\frac{1}{1-\kappa_1\kappa_2}\delta(\gamma_1+\gamma_2)+\gamma_1\right) \phi_M\int\limits_0^t|w(s)|ds.
\label{eqn:robustc_2}
\end{align}
where the first inequality comes from the fact that $|\Phi_i(u,\delta)| \leq 1$ for all $u \in [\delta,t]$. The second component $\tilde{\theta}_2$ has successive jumps at $t=j\delta$, $j\geq 1$, the solution of $\tilde{\theta}_2$ at $(t,j)$ for $ j\delta \leq t < (j+1) \delta $ and $j \geq 1$ is given by $\Tilde{\theta}_2(t,j)=\Phi_2(t,j\delta)\Tilde{\theta}_1(j\delta,j)+\gamma_2\int\limits_{j\delta}^t\Phi_2(s,j\delta)\phi(s)w(s)ds$. Hence, we have for $j \delta \leq t  < (j+1) \delta$, $j \in \mathbb{N}_{\geq 1}$,
\begin{align}
|\Tilde{\theta}_2&(t,j)| \leq |\Tilde{\theta}_1(j\delta,j)| +\gamma_2 \phi_M \int\limits_{j\delta}^t|w(s)|ds  \nonumber\\[-15pt] & \leq \left(\frac{1}{1-\kappa_1\kappa_2}\delta(\gamma_1+\gamma_2)+\gamma_1\right) \phi_M\int\limits_0^t|w(s)|ds \nonumber \\ &+ \gamma_2 \phi_M \int\limits_{j\delta}^t|w(s)|ds \nonumber \\ \label{eqn:robustcc_2} &\leq \left(\frac{1}{1-\kappa_1\kappa_2}\delta+1\right)(\gamma_1+\gamma_2)\phi_M\int\limits_0^t|w(s)|ds.
 \end{align}
 where the first and second inequalities follow from the fact that $|\Phi_2(t,j\delta)| \leq 1$, and the last inequality comes from (\ref{eqn:robustc_2}). Therefore, we can unify the bounds obtained in (\ref{eqn:robusta2}), (\ref{eqn:prop6_2}), (\ref{eqn:robustc_2}) and (\ref{eqn:robustcc_2}) for all $(t,j) \in \dom z_{\nu}$ to get (\ref{eqn:robustaa2}).
\end{proof}
{\begin{remark}
In practice, exact values of $\phi_M$ and $\eta$ may be unknown. When using conservative estimates $\hat{\phi}_M \geq \phi_M$ and $\hat{\eta} \leq \eta$, all theoretical guarantees remain valid: If conditions (\ref{eqn:inequ1})-(\ref{eqn:inequ2}) hold for $(\hat{\phi}_M, \hat{\eta})$, they hold for $(\phi_M, \eta)$.
The robustness analysis in terms of ISS/iISS bounds become more conservative and the finite-time convergence (Theorem~\ref{prop:3}) is preserved.
\end{remark}}

\begin{remark}
\label{rem:single}
    {As suggested by an anonymous reviewer, an alternative finite-time estimator consists of a single estimator with jump map $\theta^+=(I-\Phi(\delta,0))^{-1}(\theta(\delta,0)-\Phi(\delta,0)\theta(0,0))$. This estimator depends on the initial condition $\theta(0,0)$ of the estimate and is therefore not a jump map of the form (\ref{eq1}). Writing it as a hybrid system with fixed data, as required to apply the results of this paper, would require storing $\theta(0,0)$ in an additional state component with trivial flow dynamics, which yields a system of the same dimension as (\ref{eqn:fg}) whose jump map is that of (\ref{eqn:fg}) with one of the two components being constant during flows. This scheme is in fact recovered as the particular case $\gamma_2=0$ of the estimator (\ref{eqn:fg}): then $\Phi_2(\delta,0)=I$, so that (\ref{eqn:functionals}) gives $K_1=(I-\Phi_1(\delta,0))^{-1}$ and $K_2=I-K_1$, and the component $\theta_2$ remains constant during flows, playing the role of the stored value $\theta(0,0)$. The value $\gamma_2=0$ is admitted by conditions (\ref{eqn:inequ1})--(\ref{eqn:inequ2}), whereas $\gamma_1=0$ is excluded by (\ref{eqn:inequ2}).}
\end{remark}

\begin{remark}
\label{rem:gerasimov}
{The single-estimator reset in Remark~\ref{rem:single} coincides in form with the conditional-resetting algorithm recently proposed in~\cite{gerasimov2025new}, where a resetting mechanism reconstructing the unknown parameter is constructed for standard adaptation laws (gradient, memory and dynamic regressor extension). In our setting, this scheme corresponds to the limiting case of the estimator (\ref{eqn:fg}) in which the adaptation gain of one of the two estimators is set to zero, so that this estimator remains constant during flows and is updated only at jumps. Analyzing such a scheme requires following its solutions across the reset instants in ordinary time, whereas the hybrid systems framework adopted here handles the instantaneous reset of some state components while the others flow in a systematic way, and yields the well-posedness, completeness of solutions, and explicit ISS and iISS bounds established in this paper.}
\end{remark}
\section{Examples}
\label{sec:6}
This section illustrates the proposed results\footnote{Files related to simulations can be accessed via the link: \hyperlink{https://github.com/cliclab-um6p/Hybrid-Finite-Time-Parameter-Estimation.git}{https://github.com/cliclab-um6p/Hybrid-Finite-Time-Parameter-Estimation.git}}.
\subsection{Constant Unknown Parameters}
Consider the linear regression model in (\ref{eqn:estimator_measurements}), with $\theta^{*}=(1,1)$ and $\phi=[\phi_1(t), \phi_2(t)]^{\top}$, with $\phi_2(t)=4\exp(-10t)$ and $\phi_1$ is given by $\phi_1(t) = 4$ if  $t\in [0,2]$ and $\phi_1(t) = 0$ if  $t >2 $. It is clear that the regressor $\phi$ is not persistently exciting, so the classical gradient descent algorithm would not successfully estimate the parameters. Moreover, $\phi$ is square integrable, so the result presented in~\cite{aranovskiy2016performance} cannot be applied either. Finally, the system does not satisfy the conditions of algorithms $1$ and $2$ in~\cite{wang2019robust} either. It can also be seen that $\phi$ is only exciting over the interval $[0, 2]$. The numerical values for the simulations are given by: $\theta_1(0,0)=\theta_2(0,0)=(7,5)$, $\gamma_1=0.05$, $\gamma_2=0.5$, $\delta=1$. The simulation results are shown in Figure~\ref{fig:rest1}. The solid lines represent the estimation of both parameters in the noise-free case, $w(t)=0$, and the dashed lines represent the noisy case with $w(t)=6\sin(10t)$. We can see that in the noise free case, the parameter error converges to zero in $\delta=1$. In the presence of noise, the proposed hybrid algorithm ensures that the estimation error remains bounded by a function of the magnitude of the measurement noise.

\begin{figure}[!t]
	\begin{center}
		\includegraphics[scale=0.19]{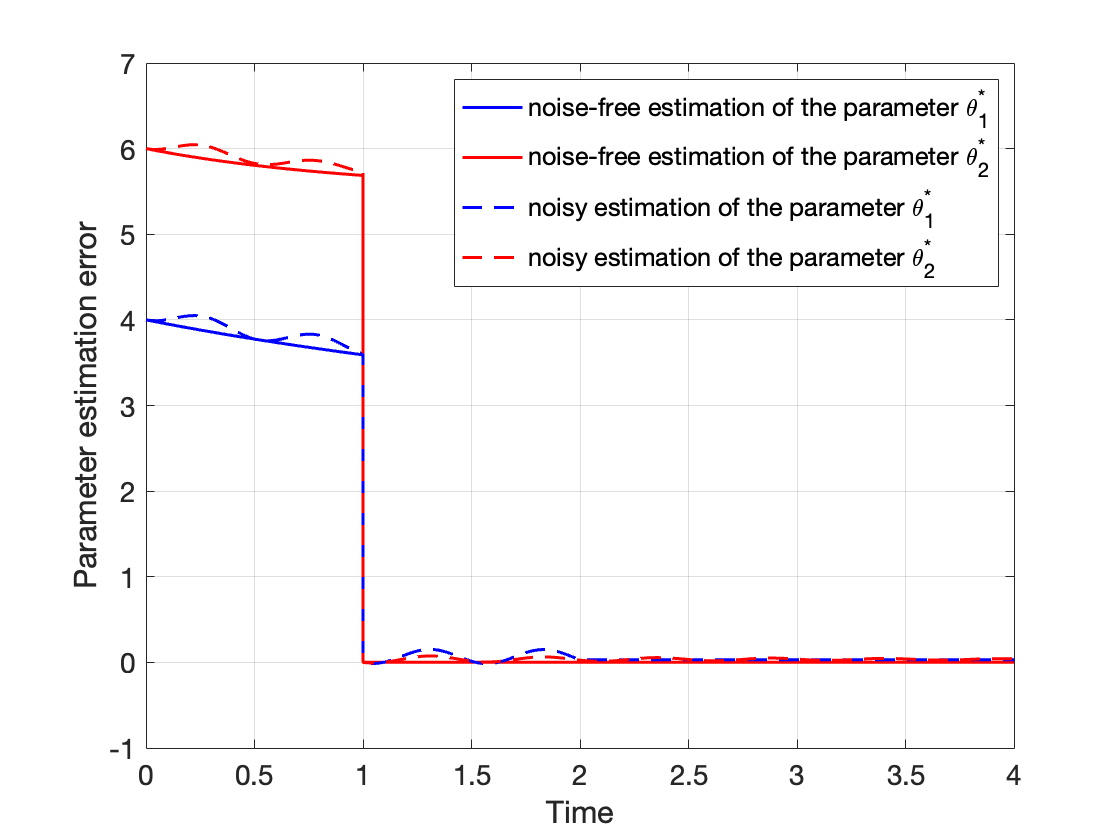}\\
	\end{center}
	\caption{Results of simulation for constant parameters}
	\label{fig:rest1}	
\end{figure}

\subsection{Time-varying Unknown Parameter}
{In this section, we show how the proposed hybrid estimator can be extended to time-varying unknown parameters $\theta^*(t) \in \mathbb{R}^n$ using standard basis function decomposition techniques~\cite{hofmann2008kernel}. Each component $\theta^*_i(t)$ is decomposed as $\theta^*_i(t) = \sum_{k=1}^M a_{k,i}\beta_{k,i}(t)$, where $a_{k,i} \in \mathbb{R}$ are unknown constant coefficients and $\beta_{k,i}(t)$ are known basis functions. Common choices include polynomial basis functions, Gaussian radial basis functions, and sigmoidal basis functions~\cite{hofmann2008kernel}. This transforms the linear regression model with time-varying parameter $y(t) = \theta^{*\top}(t)\phi(t)$ into an equivalent model with constant unknown parameters: $y(t) = \Gamma^\top\bar{\phi}(t)$, where $\Gamma = \row(A_1, \ldots, A_n) \in \mathbb{R}^{nM}$ with $A_i = (a_{1,i}, \ldots, a_{M,i})$, and the augmented regressor is $\bar{\phi}(t) = B(t)\phi(t) \in \mathbb{R}^{nM}$ with $B(t) = \diag(B_1(t), \ldots, B_n(t))$ and $B_j(t) = (\beta_{1,j}(t), \ldots, \beta_{M,j}(t))$. The hybrid estimator presented in Section~\ref{subsection:constant} can then be applied to estimate the constant parameter $\Gamma$, from which the time-varying parameter $\theta^*(t)$ is reconstructed.}

For the numerical example, we assume that $t \mapsto \theta^*(t) \in \mathbb{R}^2$ is generated using a polynomial basis and given by $\theta^*(t)=(2t-4t^2+t^3,-2t+6t-t^3)$. The regressor $\phi=[\phi_1(t), \phi_2(t)]^{\top}$, with $\phi_2(t)=4\exp(-10t)$ and $\phi_1(t) = 4$ if  $t\in [0,2]$ and $\phi_1(t) = 0$ if  $t >2 $. This model is transformed using the procedure described above to a linear regression model with a constant unknown parameter and given by $y(t)=\Gamma\bar{\phi}(t)$. In this case, the new regressor $\bar{\phi}(t)$ is given by:
$$\bar{\phi}(t)=\begin{bmatrix}
t & t^2 & t^3 & 0 & 0 & 0 \\
0 & 0 & 0 & t & t^2 & t^3
\end{bmatrix}^{\top}
\phi(t).$$
One can check that $\bar{\phi}$ is exciting over the interval $[0, 1]$. The numerical values for the simulations are given by $\theta_1(0,0)=\theta_2(0,0)=(8,-2)$,$\gamma_1=0.05$, $\gamma_2=0.5$, $\delta=1$. The simulation results are shown in Figure~\ref{fig:2}. The red line represents the time-varying parameter and the blue line represents the proposed estimator. We can see that one can exactly estimate the unknown parameter at $\delta=1$. 
\begin{figure}[!t]
	\begin{center}
		\includegraphics[scale=0.62]{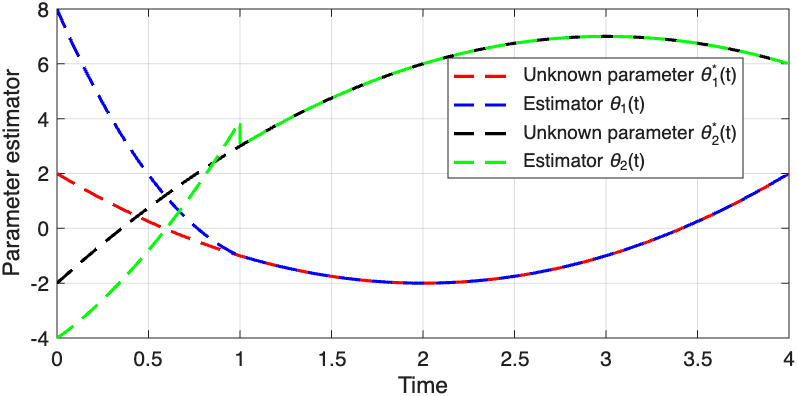}\\
	\end{center}
	\caption{Results of simulation for time-varying parameter }
	\label{fig:2}	
\end{figure}
\subsection{Scalability of the Hybrid Parameter Estimator}
{To empirically validate our computational efficiency claims, we conducted benchmarks comparing our hybrid approach with multiple DREM-based approaches~\cite{aranovskiy2016performance,aranovskiy2023drem,wang2019fixed}.
We consider the linear regression model (\ref{eqn:estimator}) with varying dimensions, where we considered sinusoidal regressors of different frequencies. Simulation parameters are set as $T_{sim} = 2\pi$ and $\delta = \pi$. All simulations were performed in Matlab on an Apple MacBook Pro (M1 Max, 64GB RAM). For each dimension $n \in \{10, 20, 50, 100, 500, 1000, 2000, 5000\}$, we measure the wall-clock time required to complete the simulation. The results are presented in Table~\ref{tab:computational_comparison}. One can see that the speedup ratio grows from 1.24$\times$ at $n=10$ to 906$\times$ at $n=5000$, demonstrating the increasing computational advantage as dimension grows.
The DREM computation time grows approximately as $O(n^3)$ with dimension, consistent with the cost of computing the determinant and the inverse at every timestep. In contrast, the proposed hybrid estimator fundamentally avoids the DREM transformation overhead by using only $O(n)$ gradient descent operations per timestep during flows, and only performing $O(n^3)$ operations (matrix inversion and computation of the state-transition matrix) exactly once per interval $\delta$ during jumps and not at every timestep. This fundamental difference makes the hybrid estimator particularly advantageous for high-dimensional systems.}
\begin{table}[htbp]
\centering
\caption{Computational Time Comparison: Hybrid vs. DREM-based Approaches (Time notation: s (seconds), m (minutes), h (hours).)}
\label{tab:computational_comparison}
\small
\begin{tabular}{|r|c|c|c|c|c|c|}
\hline
\textbf{\shortstack{Dim}} &  \textbf{\cite{aranovskiy2016performance}} & \textbf{\cite{wang2019robust}} & \textbf{\cite{aranovskiy2023drem}} & \textbf{\shortstack{Hybrid}} & \textbf{Speedup}\\
\hline
10 &  0.05s & 0.05s & 0.05s & 0.041s & 1.24$\times$ \\
20 &  0.12s & 0.12s & 0.12s & 0.010s & 12.2$\times$ \\
50 &  0.39s & 0.38s & 0.40s & 0.040s & 9.78$\times$ \\
100 &  1.08s & 1.05s & 1.11s & 0.048s & 22.5$\times$ \\
500 &  18.92s & 18.68s & 24.88s & 0.183s & 104$\times$ \\
1000 &  1m 06s & 57.53s & 1m 19s & 0.42s & 157$\times$ \\
2000 & 5m 27s & 4m 52s & 5m 48s & 1.19s & 274$\times$ \\
5000 &  1h 15m & 1h 14m & 1h 21m & 4.94s & 906$\times$ \\
\hline
\end{tabular}
\label{tab:computational_comparison}
\end{table}

\section{Conclusion}
In this paper, a robust hybrid finite time parameter estimator is proposed. When the unknown parameter is constant, we have shown that the proposed hybrid estimator allows for finite time convergence, while only requiring the regressor to be exciting on a finite interval. For the case of a piecewise constant parameter, the classical persistency of excitation condition is required to guarantee the convergence. The robustness  with respect to time-varying measurements noise is analysed using tools from hybrid systems theory. Finally, numerical examples are provided, showing the practicality of the proposed approach. In future work, we will develop finite-time parameter estimators for more complex nonlinear dynamical models.

\bibliographystyle{ieeetr}

\bibliography{ref1}

\section{Appendix: Auxiliary results}
\label{sec:App}

\begin{lemma}
\label{lem21}
Consider the linear time-varying system given by
\begin{equation}
\label{eqn:sys}
    \dot{x}=-\gamma\phi(t)\phi^{\top}(t)x
\end{equation}
where $x,\phi:t \mapsto \phi(t) \in \R^n$ and $\gamma \geq 0$. Let $\Phi(t,t_0)$ be the state transition matrix of (\ref{eqn:sys}) describing the evolution of the state between $t_0$ and $t$ with $t\geq t_0\geq 0$. If the regressor $t \mapsto \phi(t)$ is exciting over $[t_0,T]$ (i.e, $\int_{t_0}^{T}\phi(s)\phi^{\top}(s)ds\geq \eta I$ for $\eta>0$) and there exists $\phi_M \in \R_{\geq 0}$ such that $|\phi(s)| \leq \phi_M$ for all $s\in [t_0,T]$ then 
\begin{equation}
\label{eqn:bound1}
    |\Phi(T,t_0)| \leq \sqrt{\left(1-\frac{2\eta\gamma}{(1+\phi_M^2\gamma(T-t_0))^2}\right)}
\end{equation}
and each solution $t \mapsto x(t)$ to (\ref{eqn:sys}) satisfies
\begin{equation}
\label{eqn:expstab}
    |x(t)|\leq \kappa e^{-\lambda (t-t_0)}|x(t_0)|.
\end{equation}
for all $t \in [t_0,T]$, with $\kappa \geq 0$ and $$\lambda=-\frac{\log\left(1-\frac{2\eta\gamma}{(1+\phi_M^2\gamma(T-t_0))^2}\right)}{2(T-t_0)}.$$
Moreover, if $1-\phi_M^2\gamma(T-t_0)>0$ then we have 
\begin{equation}
\label{eqn:bound2}
    |\Phi^{-1}(T,t_0)| \leq \sqrt{\left(1+\frac{2\gamma\phi_M^2(T-t_0)}{(1-\phi_M^2\gamma(T-t_0))^2}\right)}
\end{equation}
\end{lemma}
\begin{proof}
If $\gamma=0$, then (\ref{eqn:sys}) gives $\Phi(T,t_0)=I$, and (\ref{eqn:bound1})--(\ref{eqn:bound2}) hold with equality, their right-hand sides being equal to $1$, while (\ref{eqn:expstab}) holds with $\kappa=1$ and $\lambda=0$. We therefore assume $\gamma>0$ in the sequel.

Let $V(x)=\frac{1}{2}|x|^2$. For the system (\ref{eqn:sys}) we have that
\begin{equation}
\label{eqn:lyap}
    \dot{V}(x,t)=-\gamma x^{\top}\phi(t)\phi^{\top}(t)x=-\gamma|\phi^{\top}(t)x|^2
\end{equation}

By integrating this expression between $t_0$ and $T$ one gets the following chain of inequalities: 
	\begin{align}\notag
	\Delta V:=&V(x(T))-V(x(t_0)) \nonumber \\ =&-\gamma\int_{t_0}^T |\phi^{\top}(s)x(s)|^2ds \nonumber \\\notag
	=& -\gamma\int_{t_0}^T \left|\phi^{\top}(s)\left(x(t_0)-\int_{t_0}^s\gamma \phi(v)\phi^{\top}(v)x(v)dv\right)\right|^2ds\\ \leq& -\frac{\rho_1}{1+\rho_1}\gamma\int_{t_0}^T \left|\phi^{\top}(s)x(t_0)\right|^2ds +\\&\rho_1\gamma\int_{t_0}^T\left|\phi^{\top}(s)\int_{t_0}^s\gamma\phi(u)\phi^{\top}(v)x(v)dv\right|^2ds \nonumber \\ \notag
	\leq& -\frac{\rho_1}{1+\rho_1}\gamma\eta|x(t_0)|^2+\\ &\rho_1\gamma^3\phi_M^4(T-t_0)^2\int_{t_0}^T |\phi^{\top}(s)x(s)|^2ds \\ \notag \leq& -2\frac{\rho_1\gamma\eta}{1+\rho_1}V(x(t_0))- \\&\rho_1\gamma^2\phi_M^4(T-t_0)^2(V(x(T)-V(x(t_0)).
	\end{align}
where the second equality comes from the fact that the solution $x(s)$ to the system (\ref{eqn:sys}) is given by
$x(s)=x(t_0)-\int_{t_0}^s\gamma\phi(v)\phi^{\top}(v)x(v)dv$, the first inequality follows from the fact that $(a-b)^2 \geq \frac{\rho_1}{1+\rho_1}a^2-\rho_1 b^2$ for all $a,b \in \mathbb{R}$ and for all $\rho_1 \geq 0$, the second inequality comes from the fact that $\phi$ is bounded and exciting over $[t_0,T]$ and the use of Cauchy-Schwartz inequality. Finally the last inequality follows from the fact that $V(x(t_0))=\frac{1}{2}|x(t_0)|^2$. Hence we have that 
\begin{equation}
\label{eqn:transi}
    V(x(T))\leq (1-\alpha)V(x(t_0)) 
\end{equation}
with 
\begin{equation}
\label{eqn:alpha}
\alpha=\frac{2\rho_1\eta\gamma}{(1+\rho_1)(1+\rho_1\phi_M^4(T-t_0)^2\gamma^2)}
\end{equation}
Since the last equality holds for all $x(t_0)\in \R^n$, we have:
\begin{align*}
 |\Phi(T,t_0)|=&\max\limits_{x(t_0)\in \R^n\setminus\{0\}} \frac{|\Phi(T,t_0)x(t_0)|}{|x(t_0)|}\\=&\max\limits_{x(t_0)\in \R^n\setminus\{0\}} \frac{\sqrt{2V(x(T))}}{\sqrt{2V(x(t_0))}} \leq \sqrt{1-\alpha}
\end{align*}

By taking $\rho_1=\frac{1}{\gamma\phi_M^2(T-t_0)}$ one gets (\ref{eqn:bound1}). Let us mention that with this choice of $\rho_1$ and using the fact that $$\eta I \leq \int_{t_0}^T\phi(v)\phi^{\top}(v)dv \leq \phi_M^2(T-t_0)I$$
we have that $\frac{2\eta\gamma}{(1+\phi_M^2\gamma(T-t_0))^2} <1$.
Moreover, (\ref{eqn:expstab}) follows directly from (\ref{eqn:transi}) and (\ref{eqn:alpha}) by choosing $$\lambda=-\frac{\log\left(1-\frac{2\eta\gamma}{(1+\phi_M^2\gamma(T-t_0))^2}\right)}{2(T-t_0)}.$$

To show that (\ref{eqn:bound2}) holds, we integrate the expression in (\ref{eqn:lyap}) between time instances $t_0$ and $t$ and we get the following chain of inequalities:

	\begin{align}\notag
	\Delta V=&V(x(T))-V(x(t_0)) \\  =&-\gamma\int_{t_0}^T |\phi^{\top}(s)x(s)|^2ds. \nonumber \\\notag
	=& -\gamma\int_{t_0}^T |\phi^{\top}(s)\left(x(T)+\int_{s}^T\gamma \phi(v)\phi^{\top}(v)x(v)dv\right)|^2dv\\  \geq& -\gamma(1+\rho_2)\int_{t_0}^T |\phi^{\top}(s)x(T)|^2ds\\&-\frac{(1+\rho_2)\gamma}{\rho_2}\int_{t_0}^T |\phi^{\top}(s)\int_{s}^T\gamma\phi(v)\phi^{\top}(v)x(v)dv|^2ds \nonumber \\  
 \\ \notag
	\geq&-\gamma(1+\rho_2)(T-t_0)\phi_M^2|x(T)|^2\\&-\frac{(1+\rho_2)\gamma^3\phi_M^4(T-t_0)^2}{\rho_2}\int_{t_0}^T |\phi^{\top}(s)x(s)|^2ds  \\ \notag \geq& -2\gamma(1+\rho_2)(T-t_0)\phi_M^2V(x(T))\\&+\frac{(1+\rho_2)\gamma^2\phi_M^4(T-t_0)^2}{\rho_2} (V(x(T)-V(x(t_0)).
	\end{align}
where the second equality comes from the fact that the solution $x(s)$ to the system (\ref{eqn:sys}) is given by
$x(s)=x(T)+\int_{s}^T\gamma\phi(v)\phi^{\top}(v)x(v)dv$, the first inequality follows from the fact that $(a+b)^2 \leq (1+\rho_2)a^2+\frac{1+\rho_2}{\rho_2} b^2$ for all $a,b \in \mathbb{R}$ and for all $\rho_2 \geq 0$, the second inequality follows from the use of Cauchy-Schwartz inequality. Finally the last inequality follows from the fact that $V(x(T))=\frac{1}{2}|x(T)|^2$.
Hence, by taking $\rho_2=\frac{\phi_M^2\gamma(T-t_0)}{1-\phi_M^2\gamma(T-t_0)}$, which is positive since $1-\phi_M^2\gamma(T-t_0)>0$, we have that
\begin{equation*}
    V(x(t_0))\leq (1+\beta)V(x(T))
\end{equation*}
with 
\begin{equation}
\label{eqn:beta}
\beta=\frac{2\gamma\phi_M^2(T-t_0)}{(1-\phi_M^2\gamma(T-t_0))^2}.   
\end{equation} 
Since the last equality holds for all $x(T)\in \R^n$, we have:
\begin{align*}
 |\Phi^{-1}(T,t_0)|=&\max\limits_{x(T)\in \R^n\setminus\{0\}} \frac{|\Phi^{-1}(T,t_0)x(T)|}{|x(T)|}\\=&\max\limits_{x(T)\in \R^n\setminus\{0\}} \frac{\sqrt{2V(x(t_0))}}{\sqrt{2V(x(T))}} \leq \sqrt{1+\beta},
 \end{align*}
and one gets (\ref{eqn:bound2}). 
\end{proof}

We also recall the following auxiliary lemma from~\cite{horn2012matrix}:

\begin{lemma}
\label{lem}
Given a matrix $A$, if $|A| <1$ then $(I-A)$ is invertible and $|(I-A)^{-1}|\leq (1-|A|)^{-1}$.
\end{lemma}

\end{document}